\documentclass[12pt, letterpaper]{article}

\usepackage{arxiv}

\usepackage[utf8]{inputenc}
\usepackage[T1]{fontenc}   % use 8-bit T1 fonts
\usepackage{amsfonts}       % blackboard math symbols

\usepackage[small]{caption}
\usepackage{graphicx}
\usepackage{amsmath}
\usepackage{amsthm}
\usepackage{amssymb}
\usepackage{booktabs}
\usepackage[algo2e,ruled,vlined]{algorithm2e}%for algorithm
\usepackage{algorithm}
\usepackage{algorithmic}
\usepackage{color}

\usepackage{natbib}
\usepackage{authblk}

\newtheorem{lemma}{Lemma}

\usepackage[pdftex,
pdfauthor={Yiheng Shen; Yuan Deng; Pingzhong Tang},
pdftitle={New Complexity Results on Coalitional Manipulation of Borda}]{hyperref}

\title{New Complexity Results on Coalitional Manipulation of Borda}
\author[1]{Yiheng Shen\thanks{shen-yh17@mails.tsinghua.edu.cn}\enspace}
\author[2]{Pingzhong Tang\thanks{kenshin@tsinghua.edu.cn}\enspace}
\author[3]{Yuan Deng\thanks{ericdy@cs.duke.edu}\enspace}
\affil[1]{Tsinghua University}
\affil[2]{Tsinghua University}
\affil[3]{Duke University}

\newtheorem{theorem}{Theorem}
\newtheorem{definition}{Definition}

\begin{document}
 \maketitle
\begin{abstract}
 The Borda voting rule is a positional scoring rule for $z$ candidates such that in each vote, the first candidate receives $z-1$ points, the second $z-2$ points and so on. The winner in the Borda rule is the candidate with highest total score. 
We study the manipulation problem of the Borda rule in a setting with two non-manipulators while one of the non-manipulator's vote is weighted. We demonstrate a sharp contrast on computational complexity depending on the weight of the non-manipulator: the problem is NP-hard when the weight is larger than $1$ while there exists an efficient algorithm to find a manipulation when the weight is at most $1$.
\end{abstract}

\section{Introduction}

Voting is a general mechanism to aggregate preferences in multi-agent systems to select a socially desirable candidate. However, Gibbrard-Sattertwhaite theorem states that all (nondictatorial) voting protocols are
{\em manipulable} \cite{gibbard1973manipulation,satterthwaite1975strategy}. Therefore, a key problem confronted by the voting mechanisms is {\em manipulation} by the voters, i.e., a voter submits a vote different from her true preference such that the outcome becomes more favorable for her. Since voting mechanisms are tailored to generate a socially desirable outcome, insincere preference report may result in an undesirable candidate to be chosen.

Although the existence of manipulation is guaranteed, the computational hardness could be the barrier for a manipulation. This paper focuses on a specific scoring rule: Borda rule. In a voting under Borda with $z$ candidates, each agent gives a list of candidates according to their preference, the candidate appearing in the $k$-th place receives a score of $z-k$. The candidate with the highest total score wins the election. We refer readers to \cite{dasgupta2008robustness} for a comprehensive survey and introduction on Borda voting and social choice in general.

Under coalitional condition, \cite{Conitzer2007When} proved that the Borda manipulation is NP-hard when votes are weighted. Reducing from the RN3DM problem~\cite{Yu2004Minimizing}, \cite{Davies2011Complexity} proved that unweighted Borda manipulation is NP-hard for a group of $2$ insincere voters (manipulators). Furthermore, \cite{Betzler2011Unweighted} proved that the problem remains NP-hard when there are more than $2$ honest voters (non-manipulators). It still remains open whether the manipulation is hard when there are only $2$ non-manipulators.

In this paper, we consider a setting with $2$ manipulators and $2$ non-manipulators. However, manipulation is in fact easy to find in this setting if all the votes are unweighted. Henceforth, we turn our attention to an environment in which one of the non-manipulator's vote is weighted by a parameter $w \in \mathbb Q_{>0}$, while the votes from other votes are unweighted, i.e., weighted by $1$. Weighted voting rules have been applied to many real world applications, such as collaborative filtering in recommender systems~\cite{pennock2000social}, and recently have been playing an important role for policy design in blockchains. For example, in the DPOS (delegated proof of stake)~\cite{larimer2014delegated}, weighted voting rule is applied for agents to elect the {\em witnesses} who are responsible for validating transactions and creating blocks, and each agent's vote is weighted by how much stakes (i.e., coins/tokens) she owns.

Our main contribution is to demonstrate a complexity boundary at $w = 1$: for $w \leq 1$, we design an efficient algorithm for the manipulators to find a manipulation, while for fixed $w > 1$, we show that it is NP-hard to find a manipulation. When $w > 1$ and $w \in \mathbb Z$, we can cast our results to the unweighted setting by considering $w+1$ non-manipulators whose votes are unweighted and among them, $w$ non-manipulators share the same vote. Therefore, we are able to recover the hardness results in unweighted settings from~\cite{Betzler2011Unweighted} for more than $2$ non-manipulators.
Based on our results, we conjecture that in a weighted Borda setting with multiple manipulators and non-manipulators, it is computationally hard for the group of manipulators to find a manipulation if and only if the total weights on the votes from manipulators are strictly less than the total weights on non-manipulators' votes. Our result is served to be the first step to understand the conjecture.

\paragraph{Related Work}

%Gibbard-Satterthwaite theorem~\citep{gibbard1973manipulation,satterthwaite1975strategy} states that any non-dictatorial voting rule with more than 3 possible outcomes must be a tactical voting for some input. As a result, almost all interesting voting rules can be manipulated. However, though the manipulation exists, it might be hard to find the manipulation. 

When there is one manipulator, \cite{bartholdi1989computational} proved that most prominent voting rules could be efficiently manipulated by Greedy manipulation. After that, there is a growing body of research on voting manipulation in the last two decades~\cite{hemaspaandra2007dichotomy,conitzer2002complexity,Conitzer2003How,russell2007complexity,xia2008generalized,xia2009complexity,xia2010scheduling,davies2010empirical,lu2012bayesian}.

%In recent years in the literature of computational social choice \cite{brandt2016handbook}, manipulation problems have been solved under almost all well-known social choice rules except Borda. Under coalitional condition, \citet{Conitzer2007When} proved that the Borda manipulation is NP-hard when votes are weighted. Reducing from the $RN3DM$ problem~\cite{Yu2004Minimizing}, \citet{Davies2011Complexity} proved that unweighted Borda manipulation is NP-hard when there are 2 manipulators. Furthermore, \citet{Betzler2011Unweighted} proved that the problem remains NP-hard when there are more than $2$ non-manipulators. 

Motivated by the results that the coalitional manipulation of the Borda rule~\cite{Conitzer2007When,Davies2011Complexity} is hard in many cases, the algorithms to calculate the manipulation for Borda rule have been proposed along the history. \cite{zuckerman2009algorithms} design an efficient algorithm for coalitional weighted manipulation, which gives a successful manipulation when given an extra manipulator with maximal weight and returns false when the manipulation is impossible. \cite{davies2010empirical} provide empirical evidence of the manipulability of Borda elections in the form of two new greedy manipulation algorithms based on intuitions from the bin-packing and multiprocessor scheduling domains, which indicate that though Borda may be hard to manipulate computationally, it provides very little defense against the coalitional manipulation. Instead of searching for an exact manipulation, practical approximation algorithms have been proposed \cite{brelsford2008approximability,Davies2011Complexity,keller2017new}.
Recently, \cite{yang2015manipulation} prove that the unweighted Borda manipulation problem with two manipulators is fixed-parameter tractable with respect to single-peaked width and \cite{yang2016exact} design the exact algorithms for both weighted and unweighted Borda manipulation. 
% \subsection{Our Contribution}
% \cite{Betzler2011Unweighted} have studied the case when the number of manipulators is larger than 2. It remains open whether the manipulation problem is hard when the number of non-manipulators is $2$. However, there is almost no study on the relation between the manipulators and non-manipulators in Borda manipulation when the number of voters are limited. We tried to fill in this blank and we believe that if the total weight of the manipulators are smaller than the total weight of the non-manipulators, the manipulation complexity should be NP-hard and this is our motivation. The situation with 2 manipulators and 2 non-manipulators may be a good starting point.

% In this paper, we first design an efficient algorithm for the unweighted situation with 2 manipulators and 2 non-manipulators. However, in weighted cases when one of the non-manipulators is weighted $w$ with $w \in \mathbb{Z}$ and $w\ge 2$, we show that the manipulating problem remains NP-hard. Finally, we generalize our result to the case where $w \in \mathbb{Q}^+$ and prove that $w=1$ is a boundary of the hardness: when $w \leq 1$, an efficient algorithm exists; otherwise, the manipulation problem is NP-hard. Furthermore, our work could generate the hardness result of \citeauthor{Betzler2011Unweighted}'s work. 

\section{Preliminaries}

In an election, there is a set of candidates $C$ and a set of voters $V$. For each voter, his vote is a total order of all the candidates in $C$. A voting protocol is a function mapping all the votes to a candidate $c \in C$, who is the winner of the election. Suppose there are $|C| = z$ candidates in total and Borda is a voting protocol in which for each vote, a candidate receives $z-1$ points if it ranks first, $z-2$ points if it ranks second, $\cdots$, $0$ if it ranks last. After adding up all the scores (number of points) in the votes, the candidate with the highest total score wins the election. If there is a tie, the protocol randomly chooses one candidate from all the candidates with the highest scores, with equal probability. 
    
In a manipulation of the voting rule, some of the voters, called the manipulators, learn the other voters' votes before they submit their own votes so that they can try to find a way to manipulate their votes to alter the outcome. We assume that the objective of the manipulators is to elect a specific candidate $c^* \in C$. A manipulation is successful if the manipulators' votes result in $c^*$'s winning probability larger than 0, i.e., no other candidate has a higher total score than $c^*$. 

In this paper, we focus on a setting with two manipulators and two non-manipulators. For convenience, we denote $C=\{c_1,c_2,$ $\cdots,$ $c_{z-1},c^*\}$ and $V=\{N_1, N_2, M_1, M_2\}$, where $N_1, N_2$ are the non-manipulators and $M_1, M_2$ are the manipulators. Moreover, we assume the vote from the non-manipulators $N_1$ is weighted by $w$ while the vote from all other voters are weighted by $1$. In other words, each point $N_1$ gives in the Borda rule is multiplied by $w$ such that the first-rank candidate receives $w(z-1)$, the second-rank candidate receives $w(z-2)$, and so on so forth. Throughout the paper, when we present the vote from $N_1$, the scores have already taken the multiplication factor $w$ into account. We denote the manipulation problem under this setting as 2-2BM.
    
To represent the votes, for the non-manipulator $N_j (1\le j \le 2)$, let $v_{ji}$ be the score that she gives to candidate $c_i (1 \le i \le z-1)$ and $v_{j*}$ be the score to candidate $c^*$. Similarly, for manipulators $M_j (1\le j\le 2)$, let $t_{ji}$ be the score that she gives to candidate $c_i (1 \le i \le z-1)$ and $t_{j*}$ be the score to candidate $c^*$. Let $V_{j}^{\mathcal N} (1\le j \le 2)$ be the set of scores given to all candidates by $N_j$ and $V_{j}^{\mathcal M} (1\le j\le 2)$ be the set of scores given to all candidates by $M_j$. Therefore, we have $V_{1}^{\mathcal N}=\{v_{11}, v_{12}, \cdots, v_{1(z-1)} , v_{1*}\}=\{0,w,2w,\cdots, (z-1)w\}$ and $V_{2}^{\mathcal N}=V_{1}^{\mathcal M}=V_{2}^{\mathcal M}=\{0,1,2,\cdots, z-1\}$.

\section{Efficient Manipulation for $w \leq 1$}

We start with the cases when $w \leq 1$, in which we design an efficient algorithm to find a manipulation.

\begin{definition}[Respective REVERSE algorithm]
Manipulator $M_j$'s vote is constructed by 
reversing $N_j$'s vote and then promoting candidate $c^*$ to the first place.
\end{definition}
\begin{theorem}\label{thm1}
Respective REVERSE algorithm can always make $c^*$ to be one of the winners when $w \leq 1$.
\end{theorem}
\begin{proof}
Let $r_j(c)$ be candidate $c$'s rank in $N_j$'s vote. After the reversion process of the algorithm, for any candidate $c$, her total score is $(1+w)\cdot z + (1-w)\cdot r_1(c) - 2 \leq 2z - 2$,
where the maximum is taken when $r_1(c) = z$. Moreover, after the promotion of candidate $c^*$, the total score of $c^*$ is $(3+w)\cdot z - w \cdot r_1(c^*) - r_2(c^*) - 2 \geq 2 z - 2$
where the minimum is taken when $r_1(c^*) = r_2(c^*) = z$. Moreover, notice that the promotion of $c^*$ does not increase the total score of any candidate other than $c^*$. Therefore, the final score of $c^*$ is at least of the total score of any other candidate, and thus, $c^*$ is one of the winners after the manipulation.
\end{proof}
% \begin{theorem}\label{thm2}
% If $c^*$ ranked last both in $N_1$ and $N_2$'s votes, Respective REVERSE algorithm can make $c^*$ win with probability $1/z$ and no other manipulation can make $c^*$'s winning probability surpass $1/z$.
% \end{theorem}
% \begin{proof}
%  We again argue by considering an equivalent way mentioned above to apply the algorithm. Since $c^*$ ranks last both in $N_1$ and $N_2$'s votes, after reversing $N_1$ and $N_2$'s votes, $c^*$ is ranked first in both $M_1$ and $M_2$'s votes. Therefore, the second promotion step does not cause any change to the ranks. As a result, all the candidates gets the same score and $c^*$'s winning probability is $1/z$. 

% Notice that the sum of all the players' scores is $2z(z-1)$, and thus, the average score of all the candidates must be $2(z-1)$, thus the only situation in which $c^*$ could possibly win is that all the candidates get $2(z-1)$ points. Therefore no other algorithms can do better than the Respective REVERSE algorithm.
% The proof for the theorem is similar to the previous one, thus we omit the proof here.
% \end{proof}

\section{Hardness Results with $w \geq 3$}
%\yuan{Combine the proofs of intergral and rational cases}

We now turn to the situation where one non-manipulator's vote is weighted by $w \geq 3$. Our reduction is based on the 2-numerical matching with target sums (2NMTS) problem.
\begin{definition}
	In a 2-numerical matching with target sums (2NM-TS) problem, given $m$ integers $2\le a_1\le a_2\le \cdots\le a_m\le 2m$ with $\sum_{i=1}^m a_i=m(m+1)$, the objective is to determine whether there exist two 1-to-$m$ permutations $p_{1}$ and $p_{2}$ such that $\forall 1\le j \le m$, $p_{1}(j)+p_{2}(j)=a_j$.
\end{definition}
\begin{lemma}
2NMTS is NP-hard.
\end{lemma}

\begin{proof}
We show that 2NMTS is a variant of RN3DM (restricted numerical 3-dimensional matching): given a multi-set $U = \{u_1$, $ \cdots,$ $ u_m\}$ of integers and an integer $e$ such that $\sum_{i=1}^m u_i+ m\cdot(m+1)=m\cdot e$, decide whether there exist two 1-to-$m$ permutations $\sigma$ and $\pi$ such that $\sigma(i)+\pi(i)+u_i=e$ for all $i$, which is proved to be NP-hard \cite{Yu2004Minimizing}. Since variable $e$ is fixed, let $a_j = e-u_j$ and therefore, the RN3DM problem is equivalent to the 2NMTS problem. Thus, 2NMTS is NP-hard. 
\end{proof}

\begin{theorem}\label{thm6}
2-2BM is NP-hard when one of the non-manipulators is weighted by $w \ge 3$ and $w\in \mathbb{Q}$.
\end{theorem}

% The intuition of the proof is to first construct the scores given by the two non-manipulators to the first $m$ candidates' votes to make the vote for the manipulators ``tight" for the $m$ candidates and then we minimize the maximal score of the remaining $z-m-1$ candidates in order to prevent them from beating $c^*$. When we construct the instance, 
Let the total number of candidates be $z=\lceil(w+2)m\rceil+1$ and we construct the non-manipulators' votes as follows: $v_{1*}=w\lceil(w-1)m\rceil+w$ and $v_{2*}=\lceil(w-2)m\rceil+a_1-2$. For any $1\le i\le m$, we set $v_{1i}=w(z-i)$ and $v_{2i}=\lfloor d_i-v_{1i}\rfloor$, where $d_i=v_{1*}+v_{2*}+2z-a_{m+1-i}$.

For the rest of the candidates, we arrange the remaining scores from $N_2$ in an increasing order $g_1$, $g_2$, $\cdots$, $g_{\lceil(w+1)m\rceil}$. Similarly, we arrange the remaining scores from $N_1$ in decreasing order $u_1$, $u_2$, $\cdots$, $u_{\lceil(w+1)m\rceil}$. For  $m+1\le i\le \lceil(w+2)m\rceil$, we set $v_{1i}=u_{i-m}$ and $v_{2i}=g_{i-m}$. We finish the proof of Theorem~\ref{thm6} by showing that the answer to the 2-2BM instance is ``Yes" if and only if the answer to the corresponding 2NMTS instance is ``Yes".

\begin{lemma} \label{lem131}
The answer to the 2-2BM instance is ``Yes" only if the answer to the corresponding 2NMTS instance is ``Yes".
\end{lemma}
\begin{proof}
By our construction of $v_{1i}$ and $v_{2i}$ for $1 \leq i \leq m$, we have $v_{1i}+v_{2i}+a_{m+1-i}-2\le v_{1*}+v_{2*}+2z-2$ and $v_{1i}+v_{2i}+a_{m+1-i}-1> v_{1*}+v_{2*}+2z-2$. Moreover, since candidate $c^*$ gets $v_{1*}+v_{2*}$ points from $N_1$ and $N_2$ and the two manipulators can give him at most $2(z-1)$ points, in order to ensure $c^*$ wins with probability larger than 0, we need to have for all $1 \leq i \leq m$, $t_{1i} + t_{2i} = a_{m+1-i} - 2$,
which implies that the answer to the instance of the 2NMTS is ``Yes".
\end{proof}
\begin{lemma}\label {lem132}
The answer to the 2-2BM instance is ``Yes" if the answer to the corresponding 2NMTS instance is ``Yes".
\end{lemma}
\newcommand{\mtwo}{\lceil(w+2)\cdot m\rceil}
\newcommand{\mone}{\lceil(w+1)\cdot m\rceil}
\begin{proof}
If the answer to the 2NMTS problem is ``Yes", the two manipulators have a way to give the scores from $\{0, \cdots, m-1\}$ to the first $m$ candidates such that the total score of any of them is at most 
\begin{align*}
    F^* &~= 2(z-1)+v_{1*}+v_{2*} \\
    &~= w\lceil (w-1)\cdot m\rceil+\lceil(w-2) m\rceil+w+a_1+2z-4 \\
    &~\geq (w^2 + 2w + 2)\cdot m + w + a_1 - 2.
\end{align*}
For the remaining scores from manipulators, we simply allocate them such that for $m+1\le i\le \lceil(w+2) m\rceil$, $t_{1i}=t_{2i}=i-1$. As for the non-manipulator $N_1$, for any $m+1\le i\le (w+2)\cdot m$,  we set $v_{1i}'=w\cdot (z-i)$. 

Recall our construction for $v_{1i}$ and $v_{2i}$ $m+1\le i\le \lceil(w+2) m\rceil$, we have $v_{1i} \leq w(z-i) \equiv v'_{1i}$ and $v_{2i} \leq i-1 \equiv v'_{2i}$.
We argue that for any $m+1\le i\le \lceil(w+2) m\rceil$, her total score is at most $F^*$ even if her scores are $v'_{1i}$ and $v'_{2i}$.

% Since the largest $m$ scores of $N_1$ are used to give candidates $\{c_1, \cdots, c_m\}$, the scores in $\{v_{1i}'\}$ are the largest possible scores left. Similarly, since the score $\mtwo$ has been used by $N_2$ to give $c_m$, the scores in $v_{2i}'$ are the largest scores possibly left. Thus, for any $i(m+1\le i\le \mtwo)$ the real scores in votes $v_{1i}$ and $v_{2i}$ cannot be larger than the scores in $v_{1i}'$ and $v_{2i}'$ respectively. 

Since all elements in $\{v_{1i}'\}$, $\{v_{2i}'\}$, $\{t_{1i}\}$ and $\{t_{2i}\}$ forms arithmetic sequences, the total scores $\{F_i\}$ should also be an arithmetic sequence for $m+1\le i \le \lceil(w+2) m\rceil$. Therefore, it is sufficient to check the total scores of the  $c_{m+1}$ and $c_{\lceil(w+2) m\rceil}$. For $c_{m+1}$, her total score $F_{m+1}$ is
\begin{align*}
F_{m+1}&= w \lceil (w+1)m\rceil+3m \leq (w^2+w+3)\cdot m+w\\
&\le (w^2+2w+2)\cdot m+w+a_1-2
\end{align*}
where the first inequalty is due to $w \lceil (w+1)m\rceil \leq w(w+1)m + w$ and the second inequality is due to $w \geq 3$ and $a_1 \geq 2$. As for $c_{\lceil(w+2) m\rceil}$, her total score $F_{\lceil(w+2) m\rceil}$ is
\begin{align*}
F_{\lceil(w+2) m\rceil}&=w+3\lceil (w+2) m\rceil-3< (3w+6)m +w \\
%&<(w^2+2w+2)\cdot m +w\\
&\le (w^2+2w+2) m+w+a_1-2
\end{align*}
where the first inequality is due to $3\lceil (w+2) m\rceil \leq 3(w+2)m+3$ and the second one is due to $w \geq 3$ and $a_1 \geq 2$.
\end{proof}
Therefore, we finish the reduction from 2NMTS to 2-2BM and show that 2-2BM is  NP-hard when one of the non-manipulators is weighted by $w \geq 3$.

\section{Case with $1 < w < 3$}

We now turn to the region with $1 < w < 3$. Recall that our construction for $w \geq 3$ relies on the condition that inequality $3w + 6 \leq w^2 + 2w + 2$ holds to upper bound $F_{\lceil (w+2) m \rceil}$. Therefore, it requires us to make a new construction to demonstrate that the manipulation is still computational hard for $1 < w < 3$.

\begin{table*}[htbp]
\centering
\begin{tabular}{|c|c|c|c|c|c|c|}
\hline
Group ID& 1 & 2 & 3 \\
\hline
$N_1 (2 \times)$ &$6m\to 5m + 1$&$5m\to 4m+ 1$ & $4m\to 3m + 1$\\
\hline 
$N_2$&$5m\to 4m+1$&$4m\to 3m+1$ & $3m\to 2m+1$\\
\hline
$M_1$&$m\to 2m-1$&$2m\to 3m-1$ & $4m\to 5m-1$\\
\hline
$M_2$&$m\to 2m-1$&$3m\to 4m-1$ & $4m\to 5m-1$\\
\hline
Group ID& 4 & 5 & 6 \\
\hline
$N_1 (2 \times)$ &$3m\to 2m+1$&$2m-1\to m$ & $m-1\to 0$\\
\hline 
$N_2$&$2m\to m+1$&$5m+1\to 6m$ & $6m+1\to 7m$\\
\hline
$M_1$&$5m\to 6m-1$&$7m-1\to 6m$ & $4m-1\to 3m$\\
\hline
$M_2$&$6m\to 7m-1$& $2m\to 3m-1$ & $5m\to 6m-1$\\
\hline
\end{tabular}
\caption{Construction of Votings on other 6m Candidates.}
\label{tab:votes}
\end{table*}

\subsection{Integral Weights}

We start with the easier case when the weight $w$ is integral, i.e., $w = 2$, and we will prove the hardness for the case when $w$ is a real number in Section~\ref{sec:real}.

\begin{theorem} \label{thm5}
2-2BM is NP-hard when one of the non-manipulators is weighted by $w = 2$.
\end{theorem}

Given an instance of 2NMTS with $a_1 \leq a_2 \leq \cdots \leq a_m$, we construct an instance of 2-2BM as follows. We set the number of all the candidates be $z=7m+1$ and the votes of non-manipulators are constructed as follows:
\begin{itemize}
	\item $v_{1*}=4m$ and $v_{2*}=m+a_1$;
	\item For any $1\le i\le m$, $v_{1i}=2(z-i)$ and $v_{2i}=z+2i-a_{m+1-i}-2m+a_1-1$.
\end{itemize}
Note that $N_2$'s vote does not contain two duplicate score since $m + a_1\le 3m < z+2i-a_{m+1-i}-2m+a_1-1\le 7m+2i-2m\le 7m$ and $z+2i-a_{m+1-i}-2m+a_1-1$ is monotonically increasing as $i$ increases. Moreover, let $P_i = 3z-a_{m+1-i}-2m+a_1-1 = v_{1i}+v_{2i}$ and we note that the current total score of candidate $c^*$ is $5m+a_1$.

Now we have allocated $m+1$ different scores for $N_1$ and $N_2$. To construct the remaining votes, we divide the remaining $6m$ candidates into 6 groups. Group $k$ contains the candidates $c_{km+1}$, $c_{km+2}$, $\cdots$, $c_{(k+1)m}$ for $1 \leq k\leq 6$. Since the remaining scores of $N_2$ are not consecutive, we consider the situation when the remaining scores of $N_2$ are the highest $6m$ scores, i.e., the scores are between $m+1$ and $7m$. We call these scores {\em virtual scores}. Now we give the method to allocate the remaining scores from the virtual scores. The detailed construction is presented in Table~\ref{tab:votes}. 

For a group containing candidates from $c_a, \cdots, c_b$, $\ell \to r$ with $\ell < r$ in the table indicates that this voter gives candidate $c_{a + j}$ $(\ell + j)$ points for all $0 \leq j \leq b - a$. If $\ell > r$, then this voter gives candidate $c_{a + j}$ $(\ell - j)$ points for all $0 \leq j \leq b - a$.
Note that the numbers shown in Table~\ref{tab:votes} are the unweighted scores, and therefore, ``$(2\times)$" means that the scores given by $N_1$ will be doubled since her vote is weighted by $w = 2$.
The only undetermined part is how $M_1$ and $M_2$ allocates their scores from $\{0, \cdots, m-1\}$ to the candidates in $\{c_1, \cdots, c_m, c^*\}$. 
We prove Theorem~\ref{thm5} by showing that the answer to the  constructed instance of 2-2BM is ``Yes'' if and only if the answer to the corresponding instance of 2NMTS is ``Yes''.

% The virtual scores facilitates our analysis 

% the remaining scores for $N_2$ are larger than the real remaining score for $N_2$. Suppose in the virtual vote, $N_2$ gives $c_i(m+1\le i \le 7m)$ the $e_i^{th}$ highest score. In the original vote, we also let $v_{2i}$ give the $c_i$ $e_i^{th}$ highest remaining score and the total score for all the candidates would not get larger, therefore we have allocated the original scores for the $6m$ remaining candidates. Since the construction is completed.
 
 \begin{lemma}\label{lem9}
 The answer to the constructed instance of 2-2BM is ``Yes'' only if the answer to the corresponding instance of 2NMTS is ``Yes''
 \end{lemma}
\begin{proof}
First of all, we have
\begin{align*}
\sum_{i=1}^m P_i&=\sum_{i=1}^m (3z-a_{m+1-i}-2m+a_1-1)\\
&=3zm-m(m+1)-2m^2+a_1m-m\\
&=m\cdot(3z+a_1-2-3m).
\end{align*}
Moreover, let $F_i$ be the sum of the final score of candidate $c_i(1\le i \le m)$. Since the two manipulators at least give these $m$ candidates $m(m-1)$ points, we obtain 
\begin{align*}
\sum_{i=1}^m F_i &\ge m\cdot(3z+a_1-2-3m)+m(m-1)\\
&=m\cdot(3z+a_1-3-2m).
\end{align*}

In addition, notice that two manipulators can give $c^*$ at most $2(z-1)$ points by both ranking him first in the vote. Therefore, the final total score $F^*$ of $c^*$ must satisfy
\begin{align*}
F^*&\le 5m+a_1+2(z-1)\\
&=z-2m-1+a_1+2(z-1)\\
&=3z+a_1-3-2m.
\end{align*}

Thus, the average score of the candidates $c_1,c_2,\cdots,c_m$ is at least $3z+a_1-3-2m$, which is equal to the highest score $c^*$ may get. In order to let $c^*$ win with non-zero probability, all $c_i(1\le i\le m)$ and $c^*$ must have the same final score $3z+a_1-3-2m$. To make $c_i$ have final score $3z + a_1 - 3 - 2m$, the manipulators must have $t_{1i} + t_{2i} = 3z+a_1-3-2m-P_i = a_{m+1-i}-2$. Recall that the equality of 
$\sum_{i=1}^m F_i = m\cdot(3z+a_1-3-2m)$ is obtained when the two manipulators give these $m$ candidates $m (m - 1)$ points, and thus, $\{t_{11} + 1, \cdots, t_{1m}+1\}$ and $\{t_{21} + 1, \cdots, t_{2m}+1\}$ must be a permutation of $\{1, \cdots, m\}$, which constitutes an instance of 2NMTS problem.
Therefore, candidate $c^*$ has a non-zero probability to win the election only if the answer to the instance of 2NMTS is ``Yes''.
\end{proof}

\begin{lemma}\label{lem10}
The answer to the constructed instance of 2-2BM is ``Yes'' if the answer to the corresponding instance of 2NMTS is ``Yes''.
\end{lemma}
\begin{proof}
When the corresponding instance of 2NMTS is ``Yes'', there is a way to construct $\{t_{11}, \cdots, t_{1m}\}$ and $\{t_{21}, \cdots, t_{2m}\}$ such that $t_{1i} + t_{2i} + P_i = F^* = 19 m + a_1$ for all $1 \leq i \leq m$. Since $a_1\ge 2$, it suffices to make sure that $F_i \le 19m+2$ for all $m+1\le i\le 7m$. %Moreover, we have $a_1 \leq m+1$ since $\sum_i a_i = m(m+1)$ and $a_1 \leq a_2 \leq \cdots \leq a_m$. 

Recall the construction presented in Table~\ref{tab:votes} and notice that in each group, the scores constitute an arithmetic sequence. We can conclude that the highest scores from group $1$ to group $6$ are $19m$, $19m$, $19m$, $19m$, $18m-1$, $17m-1$, respectively. Therefore all candidates from $c_{m+1}$ to $c_{7m}$ would not get a higher total score than $19m+2$. Thus, $c^*$ would win with non-zero probability against the virtual scores.

To construct the vote for $N_2$ without duplicated scores, let $V = [z] \setminus \{v_{2*}, v_{2i}, \cdots, v_{2m}\}$. For a candidate $c$ that is given $p$ in the virtual scores, $N_2$ gives candidate $c$ the $(7m - p + 1)$-th highest score from $V$. By this construction, it is clear that the scores candidate $c$ gets in the vote is no more than what he can get under virtual scores, which concludes our proof for the if direction.
% Since in the real vote, any candidate from $c_{m+1}$ to $c_{7m}$ cannot get higher score than in the virtual vote, thus they cannot get higher total score than $c^*$ in the real vote. So the answer to the instance of 2-2BM is ``Yes" if the answer to the instance of 2NMTS is ``Yes". We have completed the reduction from 2NMTS to 2-2BM when $w=2$. 
\end{proof}
Combining Lemma~\ref{lem9} and~\ref{lem10}, we finish the reduction and demonstrate that 2-2BM is NP-hard when $w = 2$.

\subsection{Rational Weights}
\label{sec:real}
We generalize our complexity results to the cases where $w$ is a rational number.

\begin{theorem}\label{thm8}
2-2BM is NP-hard when one of the non-manipulators is weighted by $1 < w < 3$.
\end{theorem}
\newcommand{\up}[1]{\left\lceil#1\right\rceil}
Suppose that $w=1+\epsilon$ with $0<\epsilon<2$. We prove the theorem by a reduction from a restricted form of 2NMTS:

\begin{lemma}\label{lem7}
\indent Suppose that there are $m>\max\{3, \frac{3}{\epsilon}\}$ and $m\in \mathbb{Z}$ given integers $0\le a_1\le a_2\le \cdots\le a_m\le 2(m-1)$, where $\sum_{i=1}^m s_i=m(m-1)$. Whether there exist two 0-to-$(m-1)$ permutations $p_{1}$ and $p_{2}$, s.t $\forall 1\le j \le m, p_{1}(j)+p_{2}(j)=a_j$ is an NP-hard problem. 
\end{lemma}
\begin{proof}
This is a restricted form of 2NMTS. Suppose there is a polynomial algorithm which is able to solve every instance in the restricted 2NMTS problem, we can enumerate all the finite instances of 2NMTS when $m>\max\{3, \frac{3}{\epsilon}\}$ and run the algorithm for restricted 2NMTS when $m>\max\{3, \frac{3}{\epsilon}\}$. Henceforth, by adding $1$ to each term of the two permutations, we get a polynomial algorithm for 2NMTS. Therefore we complete the reduction from restricted 2NMTS to 2NMTS, and thus, the restricted form of 2NMTS is NP-hard.  
\end{proof}
\newcommand{\sqb}[1]{\left[#1\right]}

Given an instance $a_1\le a_2\le\cdots\le a_m$ of restricted 2NMTS, we construct a corresponding instance of the 2-2BM.
When $m>\frac{3}{\epsilon}$, we set $z=D+m+1$, where $D=\lfloor(2+\epsilon) m\rfloor \cdot p+3m$ and $p$ is a sufficiently large integer. %$p$ is an integer satisfying $p\ge \max\left\{\frac{5m-a_1-(m+1)(1+\epsilon)}{[(2+\epsilon)m]-2m-2}\right.$, $\up{\frac{5m-2-\epsilon-a_1-\epsilon m}{2\epsilon m}}$, $\left.\frac{(6-\epsilon)m-a_1-\epsilon}{[(2+\epsilon)m]}\right\}$. 
\begin{figure}[!h]
    \centering
    \includegraphics[scale=1]{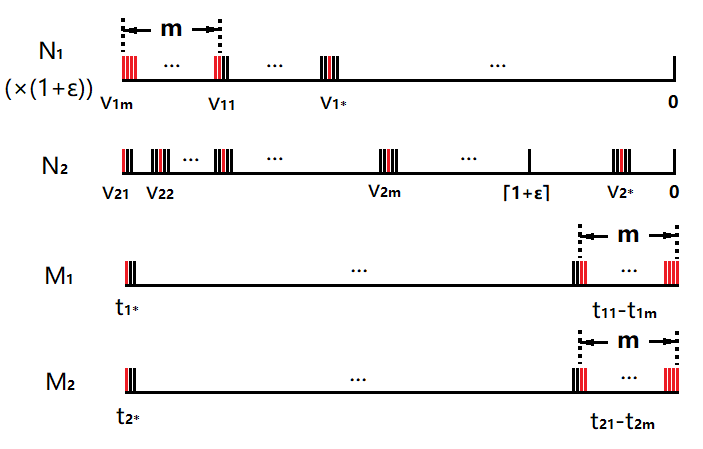}
    \caption{Vote constructions for the $\{c_1, \cdots, c_m, c^*\}$ are shown as red vertical bars}% and `` $(\times (1+\epsilon))$" means that $N_1$'s vote is weighted by $1+\epsilon$.}
    \label{fig:V3}    
\end{figure}
\paragraph{Construction for $\{c_1, \cdots, c_m, c^*\}$} For candidate $c^*$ and $c_1$, the votes from $N_1$ and $N_2$ are:
\begin{align*}
& v_{1*}=(1+\epsilon)\left\lfloor\frac{(1+\epsilon)(D+1)+a_1-D-m}{1+\epsilon}\right\rfloor,\\
& v_{2*}=\up{(1+\epsilon)(D+1)+a_1-D-m-v_{1*}},\\
& v_{11}=(1+\epsilon)(D+1) \mbox{~and~} v_{21}=D+m.
\end{align*}

First, notice that $v_{11}+v_{21}=(1+\epsilon)(D+1)+D+m$, $v_{11}+v_{21}+a_1+1>2(D+m)+ v_{1*}+v_{2*}\ge v_{11}+v_{21}+a_1$, where only the candidate with the highest rank in a manipulator's vote can receive $(D+m)$ points. Therefore, to ensure that $c^*$ wins with non-zero probability, they must give candidate $c_1$ at most $a_1$ points in total. Let $F^*=v_{1*}+v_{2*}+2(D+m)$, which is the largest possible final score of $c^*$. We construct the scores given to $c_i(2\le i \le m)$ from the non-manipulators:
\[
    v_{1i}=(1+\epsilon)(D+i) \quad\mbox{and}\quad v_{2i}=\lfloor F^*-v_{1i}-a_i\rfloor.
\]

We first argue that there is no duplicate score in our construction so far. For $N_1$, the scores he gives to candidates $c_i$ with $1 \leq i \leq m$ are monotonically increasing. Moreover, the scores he gives $c^*$ is 
$v_{1*} \le (1+\epsilon)(D+1)+a_1-D-m+1+\epsilon<(1+\epsilon)(D+1) = v_{11}$. As for $N_2$, notice that both $v_{1i}$ and $a_i$ are increasing as $i$ increases. In addition, when $i$ increases by $1$, $v_{1i}$ increases by $(1 + \epsilon) > 1$, and thus, $v_{2i}$ decreases by at least $1$. Therefore, $v_{2i}$ is strictly decreasing as $i$ increases. Finally, with sufficiently large $p$, the difference between $v_{1*}$ and $v_{11}$ is at least $D + m - a_1 - 1 - \epsilon > 1$. Therefore, there is no duplicate score in $N_2$, either.

\begin{figure}\label{fig2}[!h]
     \centering
    % \begin{adjustbox}{max width=0.9\textwidth}
    \includegraphics[scale=0.9]{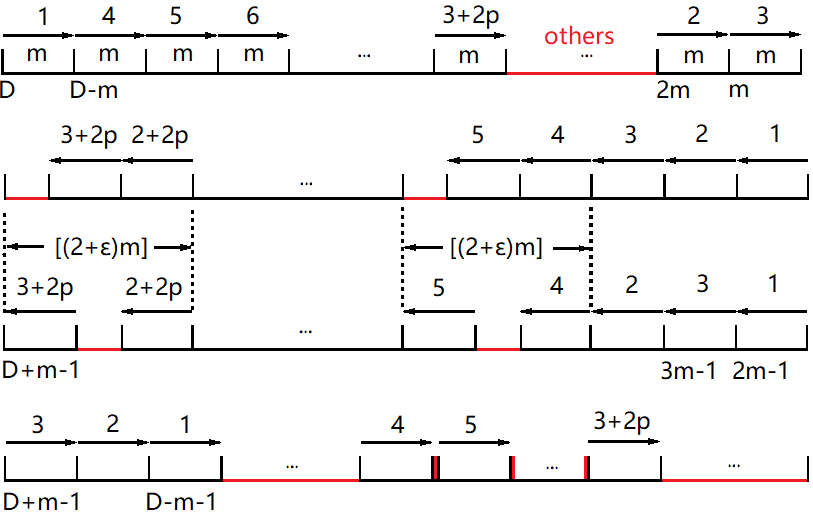}
    % \end{adjustbox}
    \caption{Vote construction for the remaining $D$ candidates. From the top to the bottom are $N_1$, $N_2$, $M_1$ and $M_2$'s votes respectively. Each segment contains $m$ scores and the labels below some of the segments show the highest score in some of the groups.}
    \label{fig:V2}    
\end{figure}
\paragraph{Construction for $(2p+3)m$ candidates in groups}
Since the used scores for both $N_1$ and $N_2$ are not consecutive, we again consider virtual scores to facilitate our analysis. Since $N_1$ has arranged all of her highest $m$ scores, in the context of virtual scores, the highest remaining $D$ scores are $\{(1+\epsilon), (1+\epsilon) \cdot 2, \cdots, (1+\epsilon) D\}$. For $N_2$, since $N_2$ has used the highest score $(D+m)$, the possible highest $D$ scores in the context of virtual scores are $\{m, m+1, \cdots, D+m-1\}$. 

%For $M_1$ and $M_2$, the original remaining $D$ scores are consecutive, so we do not need to use virtual vote technique. 

Now we construct the votes based on the virtual remaining scores as shown in Figure~\ref{fig:V2}. We will fully construct the votes from $N_1$ and $N_2$ and use all the scores except $\{0, 1, \cdots, m-1, D+m\}$ to construct votes from $M_1$ and $M_2$.

We divide the further construction into 2 steps. In the first step we select $(2p+3)$ groups of candidates. Each group contains $m$ candidates receiving consecutive scores from each voter. Each group is denoted in Figure~\ref{fig:V2} as a line segment with its group number over it. The arrows in the figure show the inner rank within each group. i.e. supposing group 4 contains $c_{a}$ to $c_{a+m-1}$, then we have $v_{1i}=(D-i+a-m)(1+\epsilon)$, $v_{2i}=t_{1i}=4m+i-a$, $t_{2i}=D+a_1-6m+\lfloor(m+1)\epsilon\rfloor-i+a$, $\forall i\in \mathbb{Z}, a\le i \le a+m-1$.

For $N_1$, we select the scores for the groups as shown in the figure. Group 2 and group 3 receive the last $2m$ scores. From group 4, the arrangement of the group is consecutive right after group 1. The red bars denote the other candidates' scores, after removing these groups of candidates.

For $N_2$ and $M_1$, the first four groups occupy the lowest $4m$ scores left for the two voters. Starting from group 5, each time we generate the scores of a group, we select $m$ scores such that (1) they are respectively $m$ points higher than the last group of the same voter; (2) they are respectively $\lfloor(1+\epsilon)m\rfloor$ points higher than the last group of the other voter. More precisely, for $N_2$ and $M_1$, we alternate in (1) and (2) between $N_2$ and $M_1$ to generate the scores. For example, in $N_2$'s voting, each member in group 5 is $m$ points larger than the member in group 4 respectively and each member in group 6 is $\lfloor(1+\epsilon)m\rfloor$ points larger than a member in group 5, while in $M_1$'s voting, each member in group 5 is $\lfloor(1+\epsilon)m\rfloor$ points larger than the member in group 4 respectively and each member in group 6 is $m$ points larger than the member in group 5. Therefore, all candidates except the candidates in group 1, 2, 3 can be partitioned into larger groups such that each group has $\lfloor(2+\epsilon)m\rfloor$ candidates. Since $D=\lfloor(2+\epsilon)m\rfloor \cdot p+3m$, there are $p$ larger groups. We can also derive that the biggest score in the group $2p+3$ of $M_1$ is $D+m-1$, which is the biggest score left for him. 

For $M_2$, the first three groups occupy the highest $3m$ scores. 
%\yiheng{Since the gap may not exist, the last inequality $\lfloor\epsilon (m+1)\rfloor + a_1 + 2 \leq 4m$ may not hold. Thus here I suppose that we could change to the following version:
%
%There may be a gap between group 1 and group 4 and the largest score in group 4 is $D+a_1-6m+\lfloor\epsilon (m+1)\rfloor \le D+a_1-6m+\epsilon (m+1)<D+a_1-6m+2(m+1)\leq D-3m+1<D-2m$, due to $a_1\leq m-1$, $m>\frac{3}{\epsilon}$ and $\epsilon<2$. $D-2m$ is the smallest score in group 1, thus group 4 and group 1 do not overlap.
%}
There is a gap between group 1 and 4, which is from $D+a_1-6m+\lfloor\epsilon (m+1)\rfloor+1$ to $D-2m-1$. Moreover, such a gap exists since we have $\lfloor\epsilon (m+1)\rfloor + a_1 + 2 \leq 4m$ due to $a_1 \leq m-1$, $m \geq 3$, and $\epsilon \leq 2$.

Starting from group 5, each member in the newly generated group is $m+1$ points lower than each member in the previous group respectively.
By our construction, the lowest score in group $2p+3$ is $s_l$, which is at most
\begin{align*}
&~ D+a_1-6m+2+\lfloor(m+1)\epsilon\rfloor-2p (m+1)\\
\ge&~ D+a_1-7m+(m+1)\epsilon-1-(m+1)(2p-1)+1\\
=&~ \lfloor(2+\epsilon)m\rfloor \cdot p+a_1-4m+(m+1)(\epsilon-2p+1) \\
=&~ (\lfloor(2+\epsilon)m\rfloor-2m-2)p+a_1-4m+(m+1)(1+\epsilon).
\end{align*}
\noindent Since $m>\frac{3}{\epsilon}$, with sufficiently large $p$, we have $s_l\ge m$.

%First we prove that the scores given in the groups in the voting of $M_2$ is feasible, i.e. they do not duplicate with the given score form $0$ to $m-1$ and the groups do not overlap. The lowest score in the group $2p+3$ is $s_l=D+a_1-6m+[(m+1)\epsilon]-(m+1)(2p-1)-m+1$. The highest score in group 4 is $s_h=D+a_1-6m+[(m+1)\epsilon]$. 

%Since $a_1\le m-1$ and $\epsilon <2$, $s_h \le D-6m+m-1+(m+1)\cdot 2 =D-3m+1<D-2m$. $D-2m$ is the lowest total score in group 1. So the lowest score satisfy the constraints and all groups do not overlap. The allocation of scores in the groups is feasible.

\paragraph{Construction for the remaining candidates} 
All that remains to construct is to arrange the scores for the candidates who are not in the $2p+3$ groups (represented as red bars in Figure~\ref{fig:V2}). There are $R=D-(2p+3)m=\lfloor(2+\epsilon)m\rfloor\cdot p-2pm$ scores and candidates left. We denote these $R$ candidates by $c_{r_1}, c_{r_2}, \cdots, c_{r_R}$. In both $N_1$ and $N_2$'s votes, we sort the remaining scores from high to low and give $c_{r_i}$ the $i$-th highest remaining score. In both $M_1$ and $M_2$'s votes, we sort the remaining scores from low to high and give $c_{r_i}$ the $i$-th lowest score.
We finish our reduction by showing that the  answer to the constructed instance of 2-2BM is “Yes” if and only if the answer to the corresponding instance of restricted form of 2NMTS is “Yes”.
\begin{lemma} \label{nlem2}
The  answer to the constructed instance of 2-2BM is “Yes” only if the answer to the corresponding instance of restricted form of 2NMTS is “Yes”.
\end{lemma}
\begin{proof}
Since $F^*-v_{1i}-a_i-1<v_{2i}$, if two manipulators give the candidate $c_i$ more than $a_i$ points in total, the total score of candidate $c_i$ will be at least $v_{1i}+v_{2i}+a_i+1>F^*-v_{1i}-a_i-1+a_i+1+v_{1i}=F^*$. Therefore,  manipulators can give $c_i$ at most $a_i$ points. 

Consider the scores given by the manipulators to candidates $\{c_1, \cdots, c_m\}$. %We have proved that the sum of the two manipulators' score on candidate $c_i(1\le i \le m)$ should not surpass $a_i$. 
Notice that the lowest $m$ scores from both manipulators sum up to $2\sum_{i=0}^{m-1}i=(m-1)\cdot m=\sum_{i=1}^{m} a_i$. Henceforth, the sum of their scores given to $c_i$ with should be exactly $a_i$ for any $1\le i\le m$. Therefore, to ensure $c^*$ can win with non-zero probability, the answer to the corresponding instance of restricted 2NMTS must be ``Yes''.
\end{proof}

\begin{lemma}\label{nlem3}
The answer to the constructed instance of 2-2BM is “Yes” if the answer to the corresponding instance of restricted form of 2NMTS is “Yes”.
\end{lemma}

When both manipulators rank $c^*$ in their first place, the total score of $c^*$ is $F^* \geq (1 + \epsilon) (D+1) + D + m + a_1$. Moreover, when the answer to the corresponding instance of restricted form of 2NMTS is “Yes”, by using scores from $\{0, \cdots, m-1\}$ only, the manipulators can ensure the total score of candidate $c_i$ with $1 \leq i \leq m$ is at most $F^*$. The remaining proof of Lemma~\ref{nlem3} is separated into Lemma~\ref{nlem0} and~\ref{nlem1}. These two lemmas together demonstrates that for any $c_i$ with $m+1 \leq i \leq D+m$, her total score is at most $F^*$.

% (Part II)
\begin{lemma}\label{nlem0}
There is no candidate in the $2p+3$ groups whose total score is more than $F^*$.
\end{lemma}

\begin{proof}
Since within each group the scores of the candidates are consecutive, forming an arithmetic sequence, the highest score in each group should appear in either the leftmost place or the rightmost place. Notice that in $N_1$ and $M_2$ the candidates are arranged from higher score to lower score and in $N_2$ and $M_1$ the candidates are arranged from lower score to higher score. In addition, the common difference of scores in each group of $N_1$'s vote is $1+\epsilon$, larger than the difference in others' votes. Therefore, the candidate with the highest total score in each group is the leftmost candidate in the group appearing in $N_1$. Let $h_i$ with $1\le i \le 2p+3$ be the highest total score among candidates in group $i$.

First we prove that the total score of candidates in group 1 is at most $F^*$: since $h_1=(1+\epsilon) D+m+m+D-2m=(2+\epsilon)D<(1+\epsilon)(D+1)+D+m+a_1\le F^*$, no candidate in group 1 has a total score higher than $c^*$.

For group 2 and group 3, notice that the candidate in group 2 have higher scores than the candidates in group 3, and therefore, we only need to consider $h_2$. We have $h_2= (1+\epsilon)(2m)+2m+3m+D-1=(7+2\epsilon)m+D-1\leq(1+\epsilon)(D+1)+D+a_1+m\le F^*$, where the second-to-the-last inequality is true by selecting a sufficiently large $p$.
For group 4, since $h_4=(1+\epsilon)(D-m)+4m+4m+D+a_1-6m-1+\lfloor(m+1)\epsilon\rfloor\le D(2+\epsilon)+m+a_1-1+\epsilon<(1+\epsilon)(D+1)+D+a_1+m\le F^*$.

For group $i$ with $i\ge 5$, first we notice that
\begin{align*}
h_i&= h_{i-1}-(1+\epsilon)m+\lfloor(1+\epsilon)m\rfloor-(m+1)+m\\
&=h_{i-1}-(1+\epsilon)m+\lfloor(1+\epsilon)m\rfloor-1\\&<h_{i-1}.
\end{align*}
Therefore, $h_i \leq h_4 \leq F^*$ for all $i \geq 5$.
\end{proof}

% (Part III)
\begin{lemma}\label{nlem1}
There is no candidate among remaining $R$ candidates outside the groups having a total score more than $F^*$.
\end{lemma}

\begin{proof}
We calculate the highest possible score of these $R$ candidates. Suppose that $N_2$ and $M_1$ give candidate $c_{r_i}$ score $x_i$ and $y_i$ respectively. Notice that all the remaining scores in $N_2$ and $M_1$ are gathered in groups, each groups has exactly $\lfloor\epsilon m\rfloor$ candidates and the distance between two adjacent groups is exactly $2m$. Since $x_i$ is the $i^{th}$ highest score left for $N_2$ and $y_i$ is the $i^{th}$ lowest score left for $M_1$, $x_i+y_i$ is a constant, which is $D+m-1+5m=D+6m-1$.

Notice that the remaining $R$ scores from $N_1$ are consecutive from $2m + 1$ to $2m + R$. Since $M_2$'s remaining scores are not consecutive, we consider a set of virtual scores such that we assume $M_2$'s remaining $R$ scores are from $D-2m-R$ to $D-2m-1$. By sorting these $R$ scores from low to high and giving $c_{r_i}$ the $i^{th}$ lowest scores, the total score of $c_{r_i}$ will not be larger than the total score in the actual vote. In the virtual scores, the scores from $N_1$ and $M_2$ form arithmetic sequences, and therefore, the highest-score candidate is the candidate who gets the highest score in $N_1$'s votes. Her total score is $(1+\epsilon)(2m+R)+(D+6m-1)+(D-2m-R)$. Recall that $R=D-(2p+3)m=\lfloor(2+\epsilon)m\rfloor\cdot p-2pm$. Therefore, we have
\begin{align*}
&~(1+\epsilon)(2m+R)+(D+6m-1)+(D-2m-R) \\
=&~ (1+\epsilon)(D - (2p+1)m) + (D + 6m - 1) + (2p + 1) m \\
\leq&~(1+\epsilon)(D+1)+a_1+D+m \\
\leq&~F^*.
\end{align*}
where the second-to-the-last inequality holds with a sufficiently large $p$.
\end{proof}

Combining Lemma~\ref{nlem0} and~\ref{nlem1}, we finish the proof of Lemma~\ref{nlem3}. Finally, combining Lemma~\ref{nlem2} and~\ref{nlem3}, we finish the proof of Theorem~\ref{thm8}. 

\section{Conclusion}
We study the problem of a Borda manipulation with 2 manipulators and 2 non-manipulators (2-2BM) with one of them weighted by $w$. We prove that when $w\le 1$, the problem is in P and when $w > 1$ and $w\in \mathbb{Q}$, the problem is NP-complete. This conclusion could generate the hardness result of \cite{Betzler2011Unweighted} when the number of non-manipulators is larger than 2 by splitting an integer-weight non-manipulator into several unweighted non-manipulators.

But the problem of Borda manipulation when there are more than 2 manipulators still remains open. This may require the hardness result of $d$-NMTS problem, which is still a demanding job. Furthermore we conjecture that the computational hard boundary lies in the total weight relation between manipulators and non-manipulators. Our work could be a good starting point.   

\newpage
\bibliographystyle{unsrtnat}  % do not change this line!
\bibliography{ref.bib} 


\end{document}